\documentclass[12pt,reqno]{amsart}

\usepackage{mathabx}
\usepackage{bbold}
\usepackage{mathrsfs}               
\usepackage{amssymb}
\usepackage{graphicx, amsmath}
\usepackage{amsfonts}
\usepackage{dsfont}
\usepackage{amssymb}
\usepackage{fancyhdr}
\usepackage{a4wide}
\usepackage{xcolor}
\usepackage{enumerate} 
\usepackage{enumitem}
\usepackage[colorlinks]{hyperref}
\usepackage{mathtools}\usepackage{mathtools}
\usepackage{scalerel}
\usepackage{scalerel}
\usepackage{bm}
\usepackage{braket}

\numberwithin{equation}{section}	 

\usepackage{enumitem}						

\renewcommand{\d}{\mathrm{d}}							

\newcommand{\C}{\mathbb{C}}			 
\newcommand{\R}{\mathbb{R}} 			

\newcommand{\calS}{\mathcal{S}}

\renewcommand{\O}{\mathcal O}

\newcommand{\tr}{\mathrm{Tr}}

\newcommand{\1}{\mathds{1}}

\newcommand{\<}{\left\langle}							
\renewcommand{\>}{\right\rangle}

\renewcommand{\leq}{\leqslant}

\renewcommand{\geq}{\geqslant}     
\newcommand{\vp}{\varphi}
\newcommand{\ve}{\varepsilon}

\renewcommand{\t}[1]{\textnormal{#1}}

\makeatletter
\newcommand*{\defeq}{\mathrel{\rlap{%
			\raisebox{0.3ex}{$\m@th\cdot$}}%
		\raisebox{-0.3ex}{$\m@th\cdot$}}%
	=}
\makeatother

\newtheorem{theorem}{Theorem} 

\newtheorem{corollary}{Corollary}[section]
\newtheorem{lemma}{Lemma}[section]

\newtheorem{definition}{Definition}
\newtheorem{assumption}{Assumption}
\theoremstyle{remark}
\newtheorem{remark}{Remark}[section]

\theoremstyle{definition}

\newcounter{listi}

\let\oldtocsection=\tocsection

\let\oldtocsubsection=\tocsubsection

\let\oldtocsubsubsection=\tocsubsubsection

\renewcommand{\tocsection}[2]{  \hspace{0em}\oldtocsection{#1}{#2}}
\renewcommand{\tocsubsection}[2]{  \hspace{2em}\oldtocsubsection{#1}{#2}}
\renewcommand{\tocsubsubsection}[2]{\hspace{2em}\oldtocsubsubsection{#1}{#2}}

\begin{document}

	\title[
Norm convergence of confined fermionic systems at zero temperature
	]
	{
Norm convergence of confined fermionic systems at zero temperature
}

	\author{Esteban C\'ardenas} 
	\address[Esteban C\'ardenas]{Department of Mathematics,
			University of Texas at Austin,
			2515 Speedway,
			Austin TX, 78712, USA}
	\email{eacardenas@utexas.edu}

	\begin{abstract}
The  semi-classical limit  of   
ground   states of large    systems of fermions 
was studied by
Fournais, Lewin and Solovej in  (Calc. Var. Partial Differ. Equ., 2018).
In particular, the authors  prove 
weak convergence 
towards   classical  states associated to the minimizers of the Thomas-Fermi functional. 
In this paper, we revisit this limit
and    show that  under additional  assumptions--and, using simple arguments--it is possible 
to prove  that strong convergence holds    
in relevant  normed spaces. 
	\end{abstract}
	
	\maketitle
%
%
%

	\section{Introduction}
In this article, we study a system of $N$   identical fermions that move  in Euclidean space $\R^d$, where $d \geq 1 $ is the spatial dimension. 
We assume that the 
particles interact through a two-body potential
$V (x-y)$, and move 
under the action of an external trap 
 $U(x)$.
  Neglecting spin variables,    the Hamiltonian of the system takes the form (in appropriate units)
	\begin{equation}
		H_N 
		\equiv 
		\sum_{ i =1 }^N  
		\hbar^2 
		(- \Delta_{x_i}) 
		+ 
		\sum_{ i = 1 }^N U (x_i ) 
		+ 
		\lambda 
		\sum_{  i < j } V(x_i - x_j )   \ . 
	\end{equation}
Here,   $\hbar>0$ plays the role of  Planck's   constant, 
and $\lambda$
is the interaction strength. 
The Hilbert space of the system corresponds  to the subspace  
$ \mathfrak h _N \equiv L^2_a( \R^{dN})$  of   functions in $L^2$, that are antisymmetric with respect to the permutation of their variables. Namely
\begin{equation}
	L^2_a( \R^{dN})
	 = 
	 \{  \Psi \in L^2(\R^{dN}) :  \Psi(x_{\sigma(1)}  , \ldots, x_{\sigma(N)  })
	  =
	    \t{sgn} \, 
	    \sigma \, \Psi (x_1 , \ldots, x_N ) \ \forall \sigma \in S_N   \} \ , 
\end{equation}
where $S_N$ stands for the permutation group of $N$ elements. 
\vspace{1mm}

In this work,  we  are interested in the case in which the external potential
$U(x)$ acts as a trap, and focus on the associated
 ground state   problem. 
 Indeed, our focus here will be on the ground state energy, 
for large number of particles $N \geq 1$, and on which the
scales of the system are semi-classical, 
and the interaction is of mean-field type. 
In other words, we consider    
	\begin{equation}
\label{EN}
		E( N ) \equiv  
		\inf  \sigma (H_N )    
		\qquad
		\t{for}
		\qquad
		\hbar 
		\equiv 
		\frac{1}{N^{1/d}}  
		\qquad
		\t{and}
		\qquad
		\lambda \equiv  \frac{1}{  N } \ . 
	\end{equation} 
	In this situation, the Fermi momentum is  of $\mathcal O ( 1 )$  and, 
	 consequently, the energy per particle is $\mathcal O (1)$ as $N \rightarrow \infty$. 
Thus,  it makes mathematical sense to analyze the asymptotics of  the ratio $E(N)/N$, 
	describing the average energy per particle in the system. 
	
	\vspace{2mm}
	
	 \subsection{Thomas-Fermi theory} 
	Heuristically, in the large $N $ limit one is able to introduce a mean-field description
	by means of the Thomas-Fermi energy \cite{Fermi,Thomas}. 
In this theory, one replaces the  $N$-particle wave function 
$\Psi_N \in \mathfrak{ h }_N$ in the energy functional, 
with its  average   over the positions 
\begin{equation}
\label{position density}
	\rho_{\Psi_N}  (x ) \equiv 
	\int_{\R^{d( N -1 )}}  	|	\Psi _N (x , x_1 , \cdots , x_ {N-1})	|^2 \d x_1 \cdots \d x_{N-1} \  , \qquad x \in \R^d \ . 
\end{equation}
In this approximation, the leading order term of $E(N)/N$ is  then expected to be 
	described in terms
	of the 
	\textit{Thomas-Fermi functional} 
	\begin{equation}
		\label{thomas fermi}
		\mathcal E
		(\rho) 
		\equiv 
		\frac{d}{d+2}
		C_{TF}
		\int_{	\R^d	} \rho(x)^{1 + \frac{2}{d}} \d x 
		+ 
		\int_{	\R^d	} U(x) \rho(x) \d x  
		+ 
		\frac{1}{2}
		\int_{\R^{2d} } \rho(x) V(x-y) \rho (y) \d x \d y 
	\end{equation}
	where $\rho(x)$ denotes the  number  density of the particles at the point $x\in \R^d$, and $C_{TF} = 4\pi^2 (d/ |\mathbb S _{d-1}|)^{2/d}$ stands for the Thomas-Fermi constant.
	The first term in \eqref{thomas fermi} 
	corresponds to the kinetic energy of the system
	after taking into account Pauli's Exclusion Principle, and is sometimes referred to as the
	quantum pressure; 
	the second and third terms on the other hand correspond to the one- and two-body potential energies, respectively. 
	The variational problem now reads
	\begin{equation}
\label{etf}
		e_{TF}
		\equiv
		\inf
		\Big\{
		\mathcal E (\rho )
		: 
		0 \leq \rho  \in L^1\cap L^{1+2/d}(\R^d),   \|	 \rho	\|_{L^1}=1 
		\Big\}  \ . 
	\end{equation}

The analysis of the relationship between    $E(N)/N$ and $e_{TF}$ 
has been the focus of extensive research. 
In particular, it has been  proven under   general assumptions on  
the interaction potentials  
that  
	\begin{equation}
\label{asymptotics}
		\lim_{N\rightarrow \infty } \frac{E(N)}{N} 
		= 
		e_{TF} \ . 
	\end{equation}
The first rigorous derivation  of \eqref{asymptotics} goes back to
the works of Lieb and Simon \cite{LiebSimon1,LiebSimon2}
for Coulomb systems, 
and has been recently revisited by Fournais, Lewin and Solovej \cite[Theorem 1.1]{Fournais2018} for  more general models. 

\vspace{2mm}

On the other hand, one may also study the problem of convergence of  the 
spatial distributions of the particles,  as $ N \rightarrow \infty  $. 
In other words,  here one is interested in proving the following convergence statement 
for the position densities 
\begin{equation}
 \lim_{ N \rightarrow \infty }\rho_{\Psi_N}  =  \rho_{TF}   \ , 
\end{equation}
where $\Psi_N$ is extracted from the many-body Hamiltonian, 
and $\rho_{TF}$ is the minimizer of $e_{TF}$. 
The first result in this direction was again proven by Lieb and Simon
 \cite[Theorem III.3]{LiebSimon2} for Coulomb systems, where convergence holds in the weak-$L^1$ sense. 
Since then, convergence has been further improved and understood. 
For instance, more recently, 
the authors in \cite{Fournais2018} have proven convergence in the weak-$L^{1+ d/2}$ sense, 
for a  significantly larger class of potentials. 
In a similar spirit, 
Gottschling and Nam \cite{GottschlingNam}
have proven convergence of the energy functionals in the sense of Gamma-convergence. 
See also the work of Nguyen 
\cite{Nguyen}, 
where convergence  is understood in the context of Weyl's semiclassical law, 
starting from the Hartree-Fock functional.

\subsection{Convergence of states}

While the analysis of  the limit
of the  spatial distribution  $\rho_{\Psi_N}$
has received considerable 
attention in the literature in the last few decades, 
the question of convergence of \textit{states} 
has only recently started to be the focus of  mathematical research. 
Note that  the interest in the former 
arises in the context of Density Functional Theory and its applications to quantum chemistry, 
where  one is  interested  in determining the \textit{static} electronic structure of many-body ground states. 



 \vspace{2mm}
 
 Let us further explain and, at the same time,  fix the notation that we use in the rest of this paper.  
Namely, let us consider a sequence of normalized states  $\Psi_N \in \mathfrak{ h }_N$ 
that satisfies 
\begin{equation}
\label{app gs}
	\< \Psi_N ,H_N , \Psi_N\>_{\mathfrak{h}_N} = E(N)
	\big(
	1 + o(1)
	\big)  \qquad N \rightarrow \infty \ . 
\end{equation} 
We shall refer to $\Psi_N$ as an \textit{approximate ground state}. 
We define 
 its  \textit{one-particle reduced density matrix} 
as 
 the trace-class operator
 on $L^2(\R^d)$
with kernel 
	\begin{equation}
	\gamma_{\Psi_N } ( x, x')
	\equiv 
	N 
	\int_{\R^{d(N-1)}} 
	\Psi_N (x , x_1 , \ldots, x_{N - 1} ) \overline{\Psi_N (x' , x_1 , \ldots, x_{N-1})}
 \ 	\d x_1 \cdots \d x_{N-1} \   \ , 
\end{equation}
for $(x,x' ) \in \R^{2d}$.  
In particular, $\tr \gamma_{\Psi_N} = N$
 and $ 0 \leq  \gamma_{\Psi_N} \leq \1 $ due to Fermi statistics. 
Because all particles are identical, 
the operator  $\gamma_{\Psi_N}$ contains all the ``one-particle"
information of the system, i.e.
   expectation values 
of    one-particle observables
can be written in terms of  traces over $\gamma_{\Psi_N}$. 
Furthermore, for weakly interacting systems, 
 higher-order reduced density matrices (containing information about particle correlations)
can be  expected to be written 
in terms of $\gamma_{\Psi_N}$
plus   an error that vanishes in the limit $N \rightarrow \infty$.\footnote{This approximation
can be justified for instance in Hartree-Fock theory, where $\Psi_N$
is an $N$-particle quasi-free state,  i.e. a Slater determinant.
Higher-order particle distributions are 
written in terms of   $\gamma_{\Psi_N}$
thanks to Wick's theorem. 
}
Thus, it is natural to study the  limit of $\gamma_{\Psi_N}$ for large $N$ in order to understand 
the physical behaviour of the system. 

\vspace{2mm}

Since our scaling regime is of semi-classical type 
 $\hbar = N^{-1/d}  $, 
we shall analyze the 
  asymptotics of $\gamma_{\Psi_N}$
by looking at its 
\textit{Wigner function} 
	\begin{equation}
\label{definition wigner}
	{f}_{\Psi_N }
	(x,p)
	\equiv  
	\int_{	\R^d	}
	\gamma_{\Psi_N} 
\bigg(
x + \frac{y}{2} , x - \frac{y}{2}
\bigg)
	e^{- i  \frac{y\cdot p }{\hbar }} \d y  \  , 
\end{equation}
where $(x,p)\in\R^{2d}$ should be regarded as varying over macroscopic scales. 
With the present definition, 
it follows from $\hbar^d N =1 $ that the following normalization holds 
\begin{equation}
\label{f norm}
\frac{1 }{ (2 \pi)^d }
	\int_{\R^{2d}}  f_{\Psi_N} (x,p) \d x \d p 
	  =  1 \ . 
\end{equation}
 While the previous equation suggests that one could regard $(2\pi)^{-d} f_{\Psi_N}(x,p)$ as 
 a classical state ({i.e.} a probability measure on $\R^{2d}$), 
 it is well-known that in general it is not positive, 
 and may take on negative values.
 On the other hand, its associated \textit{Husimi measure}, is     a well-defined   measure  on phase space.
Here, we adopt the   following   definition 
	\begin{equation}
	\label{definition husimi}
	m_{\Psi_N}  \equiv f_{\Psi_N} * \mathscr{G}_{\hbar }   
\end{equation}
where 
$	\mathscr{G}_\hbar (z ) 
\equiv 
\hbar^{-2d} \mathscr{G}_1 (  z / \hbar )
$
is a   mollifier with Gaussian profile $\mathscr{G}_1 \in L^1(\R^{2d}) $, at scale $\hbar>0$. 
In particular, the following holds   for the Husimi measure
	\begin{equation}
\label{husimi conds}
	0 \leq  m_{\Psi_N}  (x,p) \leq 1 \qquad \t{and}
	\qquad 
\frac{1 }{ (2 \pi)^d } 
	\int_{	\R^{2d}	} m_{\Psi_N} (x,p) \d x \d p =1 \ , 
\end{equation}
see for reference  \cite[Section 2]{Fournais2018}. 
The first inequality is nothing but the 
 {Pauli Exclusion Principle}, 
whereas the second identity is a normalization condition. 
Let us remark here that, at least formally,
these two functions have the same limit:  $\lim_{ N \rightarrow \infty } m_{\Psi_N} = \lim_{ N \rightarrow \infty } f_{\Psi_N}$.
Hence, it is convenient to study them simultaneously. 
\vspace{2mm}

 The    question that   one now  may ask is what 
the asymptotics of   $f_{\Psi_N}$ is, 
for large $N$. 
  In order to motivate this limit, 
  we observe that one may 
 re-write the energy functional $E(N)$ in terms of   reduced density matrices, 
 and then in terms of the  Wigner function. 
 In the limit, one    formally  obtains  the 
 \textit{Vlasov energy}
 \begin{equation}
 	\mathcal V (f) 
 	\equiv 
\frac{1}{	(2\pi)^d	}
 	\int_{	\R^{2d}	} p^2 f(x,p) \d x \d p 
 	   + 
 	   \int_{	\R^{d}	} \rho_f(x) U(x) \d x 
 	    + 
 	    \frac{1}{2}
 	    \int_{\R^{2d}}
 	    \rho_f(x) V(x-y) \rho_f(y) \d x \d y  \ , 
 \end{equation}
 where $\rho_f (x) \equiv (2\pi)^{-d } \int_{\R^d} f(x,p)\d p $ 
 is the associated position density of $f$. 
 The minimization is carried out over all  functions in  $L^1$, 
 complying with the Pauli Exclusion Principle $0 \leq f \leq 1$. 
The connection between  the 
 Vlasov  and Thomas-Fermi energies
can be understood as follows. 
Let   $\rho \in L^1(\R^d) $ be a position density and define the classical state 
 	\begin{equation}
\label{definition f in rho}
 		f_{\rho}   (x,p) \equiv \mathds{1}(|p|^2 \leq C_{TF} \rho (x)^{2/d})    , \qquad (x,p )\in \R^{2d} \ . 
 	\end{equation}
Then, there holds   $\mathcal V(f_\rho ) = \mathcal E (\rho )$.  
Note that \eqref{definition f in rho} is nothing 
but a  local version of  a Fermi-Dirac distribution at zero temperature:  
  at position $x\in \R^d$
all the electrons  momenta accomodate  in order to fill a ball 
of radius  $p_F(x)  \sim  \rho(x)^{1/d}$. 
Letting $\rho_{TF}$ be the minimizer
of   $e_{TF}$, 
we can then
expect the relevant classical state to be $f_{TF} \equiv  f_{\rho_{TF}}$.

 \vspace{2mm}

With the above notations, we are now ready to state the main question of interest in this work.
Namely: 

 \vspace{1.2mm}
\textit{Assuming that $\Psi_N$
is an approximate ground state of $H_N$, 
in what sense can we prove that its Wigner function $f_{\Psi_N}$ converges to the   state $f_{{TF}}$? }

 \subsection{Our contribution}
 	Recently,  Fournais, Lewin and Solovej  \cite[Theorem 1.2]{Fournais2018} 
	have studied this problem in great generality. 
In our context, their most notable result
is the fact that      convergence holds in a weak sense,
even without requiring the uniqueness of the minimizers.
More precisely, 
it is shown that 
there exits a probabilty measure
$\mathscr P $ over the set $\mathcal M $
of all minimizers of the Thomas-Fermi functional, such that 
the following limit holds   (up to the extraction of a subsequence)
\begin{equation}
\label{weak convergence}
	\int_{\R^{2d}} f_{ \Psi_N } (x,p) \vp(x,p) 
	\d x \d p 
 \ 	\longrightarrow 	 \ 
	\int_{\mathcal M}
	\bigg( 
	\int_{\R^{2d}}
		f_{TF} (x,p)
	 \vp(x,p) 
	\d x \d p 
	\bigg)
	\d \mathscr{P}( \rho) 
\end{equation}
for all test functions  $\vp$ with bounded first derivatives. 
 Their main ingredient is a fermionic version 
 of the de Finetti--Hewitt--Savage theorem for classical measures--see their Theorem 2.7. 
Let us     note   that their   results also include convergence for Husimi measures (in a stronger mode of convergence), as well as for  higher-order density matrices, 
and additional analysis in the unconfined case ({i.e} where some particles may be lost in the limit).

\vspace{2mm}

The work of \cite{Fournais2018} motivated 
later  studies on the semi-classical limit 
of similar fermionic 
systems. 
For instance, 
the case of \textit{positive} temperature  
has been  considered by Lewin, Madsen and Triay in \cite{Lewin2019}. 
Most notably, 
strong convergence in $L^1$ is proven for Husimi measures, 
and analogous weak convergence results are proven for Wigner functions
and higher-order density matrices.
Here,  the additional assumption of minimizers being unique is made. 
Three-dimensional systems  
in the presence of	a 	homogeneous magnetic fields were   studied by 
Fournais and Madsen \cite{SorenMadsen20}
where various    scalings  for the strength  of the magnetic field are considered.
 Girardot and Rougerie \cite{GirardotRougerie2021}
 have also analyzed the semi-classical limit of 
 fermionic anyons. 
Here,   analogous results on the weak convergence of states is   proven for  $k$-particle Husimi measures. 
Finally, let us  mention that the study of convergence of states
of large quantum-mechanical systems 
in the style presented above  was first analyzed for bosons.
See {e.g.} the work of  Lewin, Nam and Rougerie  \cite{LewinNamRougerie14}.

\vspace{2mm}
The main contribution of this article
is to revisit the     convergence \eqref{weak convergence}
assuming  
zero magnetic fields 
and 
$\hat V \geq 0 $ 
(in particular, minimizers of $e_{TF}$
are unique). 
We summarize our main results as follows.
\begin{enumerate}[leftmargin=*]
	\item In Theorem \ref{theorem 1} 
	we prove that the Husimi measure converges strongly in $L^p$ for all $ p \in [1 , \infty)$, 
	and the Wigner function converges strongly 
	with respect to certain  Fourier-based norms, 
which 	include  the negative   Sobolev space  $H^{s}$  for every  $s<0$.  
	The latter Fourier-based norms   have been previously 
	considered in the derivation of the Vlasov equation
	from many-body systems (see e.g. \cite[Theorem 2.5]{Benedikter2016}), 
	as well as in the analysis of the space homogeneous Boltzmann equation \cite{Toscani}.

	\item While  we are not able  to prove in full generality that 
	$f_{\Psi_N}$
	converges strongly to $f_{TF}$ in $L^p$ spaces, 
	in Theorem \ref{thm 2}
	we provide additional  conditions  
	under which   convergence holds. 
	Most notably, we show that $L^2$
	convergence holds   if the sequence $\Psi_N$
	consists of Slater determinants (i.e. Hartree-Fock states). 
	Additionally, we prove
	that 
	whenever the system has uniform moments in 
	phase-space, 
	then
	 strong convergence in $L^p$
	holds  for additional  values  of $ p  <  2 $. 
	 This condition on the moments 
	 is    verified for the harmonic trap $U(x) =x^2. $

	\item 
	In Theorem \ref{thm 3} we investigate
	the   convergence of higher-order distribution functions, 
	first established by the authors in \cite{Fournais2018} in a weak sense. 
	Here, we show that $L^p$ convergence can be extended to the Husimi functions to all $ p \in [ 1,\infty)$. 
	For   Slater determinants, we show that the Wigner  functions  converge strongly in $H^{s}$  for all $s<0$, 
	but \textit{not} in $L^2$. 
\end{enumerate}

 The reader may wonder if 
 upgrading the mode of convergence 
 has value beyond pure mathematical interest. 
It turns out 
this question 
has its motivation in the    \textit{quantitative}  derivation of effective equations  for 
large systems of interacting fermions.
Indeed, here 
one   considers  an externally prepared system of inital data, 
that converges with respect to some norm, as  $ N \rightarrow \infty$. 
Subsequently, one wishes to prove (and, quantify)
that convergence is propagated at later times 
along the solutions of the equations under consideration; 
see e.g
\cite{Benedikter2014,Benedikter2016,Benedikter2016-2,FrestaPortaSchlein2023} and the references therein. 
%
%
%
%
%
%
In this context,  our main result  
  provides  examples of initial data for the  derivation of effective equations, 
which converges with respect to appropriate norms. 
Most importantly, these examples stem from states at zero temperature
and include the orthogonal projections (i.e. the  Slater determinants). 
In particular, one must be careful with the choice of norm  in this situation, 
since regularity of the limiting function is not abundant--the functions 
$f_\rho$ introduced in \eqref{definition f in rho}
are  only of bounded variation.

	 \section{Main results}

In this section, we state the main results of this article. 
Namely, Theorem \ref{theorem 1} and \ref{thm 2}. 
First, we fix the assumptions
on the potentials that we work with, as 
well as fixing the notation to be used throughout this paper. 

\subsection{Assumptions, definitions and notations}

We will work with the following set of assumptions. 
These are taken from \cite{Fournais2018}
up to the following modifications.
First,  we take $A \equiv 0$.
Second, we assume   
$\hat V \geq 0$. 
\begin{assumption}
	\label{assumption 1}
	$U : \R^d \rightarrow \R$
	and
	$V : \R^d \rightarrow \R$
	satisfy the following conditions. 
	
	\begin{enumerate}
		\item   
	$
			U_- \in 
			L^{1 + d/2}	(\R^d )
			+ L_\ve^\infty (\R^d)
			\quad
			\t{and}
			\quad 
			U_+ \in L_{\t{loc}}^1(\R^d) , 
			\ \lim_{|x| \rightarrow \infty} 
			U_+ (x) = + \infty \ . 
$

		\item $V$ is even,  $V  \in 	L^{1 + d/2}	(\R^d )
		+ L_\ve ^\infty (\R^d)$, and $\hat V \geq 0 . $
	\end{enumerate}
\end{assumption}

\begin{remark}
Given $p,q\in[1,\infty]$,  we write $ U \in L^p + L^q_\ve$ 
if the following is satisfied: 
for all $\ve>0$ there exists $U_1 \in L^p$
and 
$U_2 \in L^q$ 
such that 
	$U = U_1 + U_2$
and $\|	 U_2\|_{L^q} \leq \ve$. 
\end{remark}

\begin{remark}
Under these conditions, there always exists a minimizer of the Thomas-Fermi funtional.
In addition, thanks to  $\hat V \geq0$, 
the 
	functional $\mathcal E (\rho)$
	is strictly convex in $\rho$.
Hence, the minimizer is unique and is denoted by $\rho_{TF}$. 
\end{remark}

\begin{remark}
The additional conditions $A= 0 $
and 
$\hat V \geq  0 $
can be relaxed.
Namely,  as long as the Thomas-Fermi functional has a unique minimizer, then 
  Theorem \ref{theorem 1} and \ref{thm 2} still hold. 
These have only been included here for notational convenience
and simplicity of the exposition.
The only result that makes explicit use of these additional conditions
  is Lemma \ref{lemma moments} (see Remark \ref{remark ho} below), 
  although we believe 
  the result of the lemma   is still  true  for potentials $A$ and
$\hat V$
verifying  the  more general  assumptions considered in 
 \cite{Fournais2018}. 
\end{remark}

\vspace{2mm }
  
\textit{States}. 
Let $  (   \Psi_N ) _{N \geq 1} $ 
be a sequence of approximate ground states satisfying 
 \eqref{app gs}. 
We   let  $\gamma_{\Psi_N}$
be its one-particle reduced density matrix, 
$f_{\Psi_N}$ its  Wigner  function   and  $m_{\Psi_N}$  its  Husimi  measure.
In order to simplify  the notation, here and in the sequel  we  write for all $ N \geq 2 $
\begin{equation}
\label{states}
	f_N \equiv f_{\Psi_N}   \   , 
	\qquad 
	m_N \equiv m_{ \Psi_N} , 
	\qquad 
	\gamma_N\equiv \gamma_{\Psi_N} \   , 
	\qquad
	\t{and}
	\qquad 
	f \equiv f_{\rho_{TF}} \  , 
\end{equation}
where $\rho_{TF}$ is the unique minimizer of $e_{TF}$ and $f$ is defined through 
 \eqref{definition f in rho}.

\vspace{3mm}

\textit{Weak convergence}.
Let us denote $\<\vp, g\> = \int_{\R^{2d}} \vp (z) g(z) \d z $. 
Then, with the above notation, it follows from  
\cite[Theorem 1.2]{Fournais2018}
and uniqueness of the Thomas-Fermi minimizer 
that  
\begin{equation}
	\label{convergence minimizers}
	\lim_{N\rightarrow \infty }
	\< \vp , f_{ N } \>
	= 
	\< \vp ,  f   \> 
\end{equation}
for all  test functions 
$\vp $ with  $\vp , \, \nabla \vp \in L^\infty (\R^{2d}). $ 
On the other hand, for the Husimi measures  
\begin{equation}
	\label{convergence minimizers husimi}
	\lim_{N\rightarrow \infty }
	\< \vp , m_{N} \>
	= 
	\< \vp ,  f   \> 
\end{equation}
for all test functions $   \vp \in   L^1(\R^{2d})  +L^\infty (\R^{2d})$.  
Let us note that the original result in \cite{Fournais2018}
includes extraction of a subsequence via a compactness argument. 
Since the limit here is unique and  independent  of the subsequence, 
convergence  holds for the  whole sequence. 
 
\vspace{3mm}

\textit{Wigner transform}.
We   denote by $W^\hbar : L^2_{x,x'}(\R^{2d}) \rightarrow L^2_{x,p} (\R^{2d})$ 
the linear map 
\begin{equation}
\label{wigner def}
	W^\hbar[\gamma] (x,p)
	 \equiv 
	\int_{	\R^{d}	} \gamma \Big(
	x - \frac{y}{2} , x + \frac{y}{2}
	\Big) e^{ - \frac{i p \cdot y}{\hbar}} 
	\d y  \ , \qquad (x,p)\in \R^{2d}
\end{equation}
which we refer to as   the Wigner transform. 
We abuse notation 
and   identify Hilbert-Schmidt operators with their $L^2$ kernels. 
In particular,  
$f_N = W^\hbar [\gamma_N]$. 
The map  $W^\hbar$ is  an $L^2$-isomorphism, 
and the   normalization is chosen so that
$\|	 W^{\hbar}[ \gamma ]	\|_{L^2} = (2\pi\hbar)^{d/2} \| \gamma	\|_{L^2}$.
See e.g
\cite[Proposition 13]{Combescure2012}.

\vspace{3mm}

\textit{Fourier-based norms}. 
Let $n \in \mathbb N$.
Throughout this article, we let   $\<\zeta \>\equiv (1 + \zeta^2)^\frac{1}{2}$ be 
the Japanese bracket, 
and    $\hat g (\zeta) 
= 
(2\pi)^{-n/2}  
\int_{\R^{n}}
e^{ - i \zeta \cdot z } g(z) \d  z $
the Fourier transform of $g$.

\begin{definition}
Let $s \in\R $ and $q \in [1,\infty]$. 
We denote by $H_{s,q}(\R^{n})$
the space of tempered   
distributions $g \in \calS ' (\R^{n} )$
whose Fourier transform is  regular   
$\hat g \in L^1_{\t{loc}}(\R^n )$, 
and for which the   norm 
\begin{equation}
	\label{sobolev norm}
	|g |_{s,q} 
	\equiv 
	\|	 \< \zeta \>^{s} \hat g	\|_{L^q(\R^n )}
\end{equation}
is finite. 	
\end{definition}

We will use these spaces only in the case of negative order $s<0$. 

\begin{remark}
	A few comments are in order regarding the Fourier-based spaces $H_{s,q}$. 
	\begin{enumerate}[leftmargin=*]
		\item 
For all $s<  0 $
there holds 
		$	|g|_{s,2} \leq \|g	\|_{L^2}$
		and 
		$	|g|_{s, \infty} \leq \|	g\|_{L^1}  $. 
		
		\item 
		For $q=2$, $ H_{s,2} = H^{s }$
		is the standard   Sobolev space of order $ s \in \R  $.

		\item 
For $q=\infty$,       the spaces 
		$H_{s, \infty}$  are quite useful in the study of fermionic systems. 
		Namely,  for the Fourier transform of $f_N$
		the following formula holds, 
		sometimes known as   \textit{Groenewold's formula}: 
		\begin{equation}
			\label{groenewold}
			\hat f_N ( \zeta  ) 
			= \frac{1}{N}
			\t{Tr}  \Big[
			\O_{\xi ,\eta}  \gamma_{ N} 
			\Big]
			\ , \qquad  \zeta = (\xi, \eta) \in \R^{2d} \  . 
		\end{equation}
		Here,  
		$\O_{\xi , \eta  } \equiv  \exp(  i  \xi \cdot \hat x+ i \eta \cdot  \hat p  )$
		is  a  semi-classical observable, 
		with $\hat x$ and  $\hat p \equiv - i \hbar \nabla_x$
		the standard position and momentum observables on $L^2(\R^d).$
		In particular, while the $L^1$ norm of $f_N$ cannot be controlled
		by trace norms of $\gamma_N$,  
		it follows easily from  \eqref{groenewold} that 
		\begin{equation}
			\label{bounds L inf}
			| 		\hat f_N( \zeta ) |
			\  \leq  \ 
			\frac{1}{N} \|	 \O_{\xi , \eta  }	\|_{B(L^2)}  \,  \t{Tr}  \gamma_ {N}
			\   \leq  \ 
			1  \ , \qquad
			\zeta  = 	(\xi, \eta) \in \R^{2d} \ . 
		\end{equation}
		Consequently, $|f|_{s,\infty} \leq1$ for all $s \leq 0$. These  uniform bounds 
		will replace    the possible  lack of $L^1$ boundedness in our analysis. 	
	\end{enumerate}
\end{remark}

%
%
%
%

	\textit{Slater determinants}.
We say that  $\Phi \in \mathfrak{ h }_N$ is a Slater determinant 
	if there exists an orthonormal set $(\vp_i)_{i=1}^N \subset L^2(\R^d)$
	such that
	\begin{equation}
\label{slater}
\Phi  (x_1 , \ldots, x_N) = \frac{1}{\sqrt{ N!} }
		\det_{1  \leq i,j \leq N } 
		\Big[
		\vp_i ( x_j )
		\Big] \  , \qquad (x_1, \ldots, x_N) \in \R^{dN}  , 
	\end{equation}
 and we write $\Phi  = \vp_1 \wedge \cdots \wedge \vp_N$. 
	In particular,   $\Phi $ is a Slater determinant  if
	and only if its one-particle reduced density matrix is an orthogonal projection, {i.e.} if 
$
		\gamma_N^2 = \gamma_N   . 
$
%
 \subsection{Main results}
	
 We are now ready to state our main result, which is the content of the next theorem. Here, we make no additional assumptions
 on the sequence $\Psi_N$ under consideration. 
 The proof, given in  Section \ref{section proof main results},
 uses only  elementary $L^p$ inequalities
  and contains the main idea of the paper.
 Namely, that weak convergence can be improved to strong convergence, 
 using  the fact that the limiting function 
 solves the equation $f^2 = f$.

	\begin{theorem}[Norm convergence]
		\label{theorem 1}
Let $f_N$,  $m_N$ and $f$ be as in \eqref{states}, and let   Assumption \ref{assumption 1} hold. 
		Then, the  following statements are true.  
		\begin{enumerate}[leftmargin=1cm]
			\item 
			The sequence $m_N$ converges  to $f$ strongly in $L^p(\R^{2d})$ for all 
		$ p \in [1,\infty)$.

			\item The sequence $f_N$ converges to $f$ strongly in $H_{s,q}(\R^{2d})$
			 for all $q \in [2,\infty]$ and $s<0$. 
	\end{enumerate}	
	\end{theorem}

\begin{remark}[$k$-particle functions]
Originally, 
the authors in \cite{Fournais2018}
proved that  \eqref{convergence minimizers} and  \eqref{convergence minimizers husimi} 
  hold for higher-order Wigner and Husimi functions.
Hence, the reader may   wonder if the content of Theorem \ref{theorem 1} 
extend  to these functions. 
The answer to this question is rather subtle
and is explained further in Section \ref{section higher}. 
\end{remark}

\begin{remark}[Positive temperature]
Note  that the equation $f = f^2$ holds \textit{only} at zero temperature, 
	because $f$ is a characteristic function of the phase-space.
	 In particular, the proof of Theorem \ref{theorem 1}
	does not apply directly to the positive temperature setting   \cite{Lewin2019}. 
\end{remark}

\begin{remark}[Position densities]
Let $\rho_{TF} $ be the unique minimizer of the Thomas-Fermi functional $e_{TF}$. 
Then,  Theorem \ref{theorem 1}  implies that 
	\begin{equation}
		\int_{	\R^{d}	} m_N( \cdot , p ) \d p \rightarrow \rho_{TF} \qquad 
		\t{strongly in } L^1(\R^d) \ . 
	\end{equation}
On the other hand, for the position density 
$\rho_N (x)   
= \int_{	\R^{d}	} f_N(x,p) \d p $
defined initially  in terms of $\Psi_N$ 
as \eqref{position density}, we have
\begin{equation}
\label{rho}
	\rho_N \rightarrow \rho_{TF}
	\qquad 
	\t{strongly in } H_{s, \infty }	(\R^d)   
\end{equation} 
for all $s< 0$. This follows from  
$ \hat \rho_{N} (\xi) = \hat f_N(\xi , 0 )$
and  $\<  \xi \>^{-s} \leq \< \zeta \>^{-s}$
for $\zeta  = (\xi,\eta)\in \R^{2d}.$ 
\end{remark}

A natural question is if the Wigner function   $f_N$
converges to $f$
in any   $L^p$ space.
We are not able to prove this statement in the most general case, 
but we are able to show convergence under additional assumptions
on the sequence $\Psi_N$. 
The first result in this  direction establishes $L^p$ convergence 
whenever the sequence $f_N$ is known to verify 
a uniform \textit{smoothness} assumption. We record this in the following corollary. 

\begin{corollary}\label{corollary}
	Under the same conditions of Theorem \ref{theorem 1}, 
	assume additionally that 
there exist
	$r>0$, 
	$p\in[1 ,\infty)$
	and $q \in [1,\infty ]$
	such that   
	\begin{equation}
		\label{sobolev type}
		\lim_{N \rightarrow \infty }
		\Big\|
		\frac{\|	 f_N  - f_N ( \bullet + \hbar z  )		\|_{L^p}}{|z|^r}
		\Big\|_{L^q(d z )}
		=0 \ , 
	\end{equation}
	then   $f_N$ converges to   $f$ strongly  in $L^p(\R^{2d})$. 
\end{corollary}

	\begin{remark}
	The condition given in \eqref{sobolev type} 
	can be interpreted as  smoothness of the sequence
	that 	is \textit{uniform} with respect to the number of particles  $N.$
	For illustration, 
	we consider    two   situations  have been   considered in the literature. 
	\begin{enumerate}[leftmargin=*]
		\item 
			  	Assume that 
		$\sup_{N \geq 1} \|	 \nabla f_N	\|_{L^1} <\infty$. 
		The Taylor formula 
		$ f (z + \hbar z' ) - f(z) = \hbar z ' \int_0^1 \nabla f(z + t \hbar z' ) \d t $  
		and a change of variables 
		$z \mapsto z - t \hbar z' $
		implies  
		for $r =1$, $p=1$ and $q = \infty$
		\begin{equation}
			\label{sobolev type 2}
			\sup_{z'\in \R^{2d}}
			\frac{1}{|z'|} \int_{	\R^{2d}	} 
			|f_{N}  ( z  + \hbar z'   )  - f_N (z)  | \d z  \leq 
			\hbar \|	 \nabla f_N	\|_{L^1}     , 
		\end{equation}
		which yields $L^1$ convergence. 
The condition 
		$\sup_N \|	\nabla f_N\|_{L^1}< \infty$
		has been previously considered   in the derivation of mean-field dynamics
		for   fermionic systems,
see e.g 		\cite[Theorem 2.5]{Benedikter2016}
and the  remark after the theorem. 
		
		\item 
		Given $z_0 = (x_0,p_0)\in \R^{2d}$,  it is possible to verify the following  relation between translations in   phase-space distributions, and on quantum density matrices 
\begin{equation}
	\|	 f_N	-			f_N (\bullet + \hbar z_0)	\|_{L^2}
	 = 		(2\pi\hbar)^{d/2}	\|	 [	\O_{ p_0 ,  -x_0	} , \gamma_N	]			\|_{HS} \ . 
\end{equation}
In particular, if $\gamma_N$  satisfies the commutator estimates 
\begin{equation}
\label{commutator}
	\|	 [	\O_{  p,x 	} , \gamma_N	]			\|_{HS}^2 
	 \leq C N  \hbar  |z|  \ , \qquad \forall z  = (x,p)\in \R^{2d}
\end{equation}
then   $f_N$
satisfies \eqref{sobolev type} for $p=2$, $q=\infty$
and $r=1/2$, 
and yields $L^2$ convergence. 
Estimates of the form \eqref{commutator}
(as well as their stronger \textit{trace-class} variants)
arise in practice
for initial data  in the derivation of effective dynamics
for fermion systems 
\cite{Benedikter2014,Benedikter2016,Benedikter2016-2,FrestaPortaSchlein2023},
and  have been shown to be satisfied in a few special cases 
 \cite{Benedikter2023,Soren19}.

	\end{enumerate}

	\vspace{2mm}

\end{remark}

In practice,  verifying the smoothness condition \eqref{sobolev type}
is a  challenging task. 
Our next result brings an alternative.
Namely, if one assumes   that $\Psi_N$ is a Slater determinant, then $L^2$ convergence is immediate.
Further, we show that uniform control on the \textit{moments}
of the system yields  $L^p$ convergence for smaller values of $p$.  
 \vspace{2mm}

\begin{theorem}[Slater determinants]
	\label{thm 2}
Under the same conditions of Theorem \ref{theorem 1}, assume additionally  that   $\Psi_N$ is a Slater determinant for all $N\geq1.$
 Then,  the following holds. 

		\begin{enumerate}[leftmargin=1cm]
				\item 
The sequence	$f_N$  converges to  $f$  strongly in $L^2(\R^{2d})$.

			\item (Moments) 
			If there exists
  $ m >0 $ such that 
			\begin{equation}
\label{moments}
\sup_{N \geq 1}
				\|	 (|x|+ |p|)^m f_N  	\|_{L^2}  < \infty  , 
			\end{equation}
then $f_N$ converges to $f$ strongly in $L^p(\R^{2d})$ 
for all 
$p \in [1,2] \cap (  \frac{2}{1 + 2m/d} , 2  		]$.

		\end{enumerate}  
	\end{theorem}

\begin{remark}
[The harmonic trap]
\label{remark ho}
For the external potential  $U(x) = x^2$
we verify
in Section  \ref{section harmonic} 
that the bounds  on the moments  \eqref{moments} hold
for $ m = 1 $, 
under Assumption \ref{assumption 1}. 
In particular, this yields   
\begin{equation}
\label{convg}
	f_N \rightarrow f \qquad 
	\t{strongly in }  L^p(\R^{2d}) \quad \t{for } 
 \begin{cases}
	p\in [1,2]    &\t{ if } d =1  \ ,  \\ 
 	p \in (1,2]   &\t{ if }  d= 2  \ , \\ 
  p \in (6/5,2)     &\t{ if } d =3  \ , 
 \end{cases}
\end{equation}
for Slater determinants. 
We would like to note that for $d=3$
  our assumptions include systems that interact through repulsive   Coulomb potentials
$V(x) = \lambda  |x|^{-1}$ with $\lambda > 0 . $
\end{remark}

	\begin{remark}
		The $L^2$ convergence result  \textit{does not} apply to  the 
		higher-order Wigner functions $f_N^{(k)}$ (see \eqref{def wigner k} for a definition). 
		In particular, we prove in Section \ref{section higher} 
		that for Slater determinants  
		\begin{equation}
			\liminf_{N\rightarrow \infty}
			\|	f_N^{(k)}	-	 f^{\otimes k}	 	\|_{L^2} \geq (2\pi)^{\frac{dk}{2}}
			\Big(
			\sqrt{k!} -1 
			\Big)	\qquad \forall k \geq 2 
		\end{equation}
		even though convergence $f_N^{(k)} \rightarrow f^{\otimes k }$
		holds in   $H^{s}(\R^{2dk})$ for all $s< 0$. 
	\end{remark}

The last of our main results is  the following interesting corollary. 
Here,  we derive an equivalent  formulation 
 of convergence in $L^p$ norm for Slater determinants. 
 
 \begin{corollary}\label{corollary 2}
Assume that   $\Psi_N$ is a Slater determinant for all $N\geq1$, 
 	and let $ p \in [1,2]$. 
 	Then, $f_N$ converges to $f$ in $L^p(\R^{2d})$ if and only if 
 	\begin{equation}
 		\lim_{ N \rightarrow \infty }	\|	 f_N	\|_{L^p} =	 (2\pi)^{d/p}	 \ . 
 	\end{equation}
 \end{corollary}


\section{Proof of the main results}
In this section we prove our main results, Theorem \ref{theorem 1} and Theorem \ref{thm 2}.

\label{section proof main results}
\subsection{Proof of Theorem \ref{theorem 1}}

The heart  of the   proof of  Theorem \ref{theorem 1}  has as a starting point the 
modes of convergences \eqref{convergence minimizers}
and
\eqref{convergence minimizers husimi}, 
proven in \cite{Fournais2018}.
We establish  strong convergence  by making heavy use
of the identity $f^2 = f $, 
which follows from the fact that $f$ is a characteristic function. 
In particular, 
it holds that $f (1-f) = 0$, i.e.
 we regard   $f$ and $1-f$  as    orthogonal projections.

\vspace{2mm}

It turns out that the proof  can be formulated
 using  rather general arguments, 
and is independent of the fact that $f_N$ is extracted from an $N$-body quantum system. 
For transparency, we prove the following 
abstract lemma
which may be of interest in its own right. 
To the authors best knowledge,  this   result 
(as well as its application to fermion systems) is new.
Essentially, it states that an appropiate notion of weak convergence
towards a characteristic function
can always be upgraded to strong convergence in $L^p$. 

\begin{lemma}\label{lemma abstract}
	Let $(X, \mathscr{F}, \mu)$ be a measure space. 
	Consider a  sequence $(F_n)_{n=1}^\infty$ of non-negative functions $F_n $
	in $L^1 (X ) \cap L^\infty (X)$ 
	that satisfies 
	\begin{equation}
		\label{cond Fn}
		0 \leq F_n \leq 1  \ \ \mu\t{-a.e }
		\qquad \t{and}\qquad 
		\limsup_{N \rightarrow \infty}  \int_X F_n \d \mu =1 \ .
	\end{equation}
	Let $F  \in L^1 (X ) \cap L^\infty (X)$ 
	be a non-negative function
	satisfying 
	\begin{equation}
		\label{cond F}
		0 \leq F \leq 1  \  \ \mu\t{-a.e }
		\qquad \t{and}\qquad 
		\int_X F \d \mu =1 \  , 
	\end{equation}
	and assume that for all $\phi \in L^1(X ) + L^\infty(X )$
	\begin{equation}
		\lim_{ n \rightarrow \infty }  \<  \phi ,F_n \>   =  \<  \phi ,F \>   \ . 
	\end{equation}
	Then, if $F = F^2 \ \mu $-a.e., there holds  
	\begin{equation}
		\lim_{ n \rightarrow \infty }	 \|	 F_n - F	\|_{L^p} =0 
		\qquad \forall p \in [1 , \infty) \ . 
	\end{equation}
\end{lemma}

\begin{proof}[Proof of Lemma \ref{lemma abstract}]
	First, we prove the case $p=2$.
	Then, we prove the case $p=1$.
	Every other case $ p \in (1,\infty)$
	then follows from the $p=1$ case, and   uniform $L^\infty$ bounds.

	\vspace{2mm}
	
	Let us prove the $p=2 $ case. 
	Since the $L^2$ norm is generated by an inner product, we find that 
	\begin{equation}
		\label{eq 1}
		\|	 F_n - F  	\|_{L^2}^2
		= 
		\|	 F 	\|_{L^2}^2
		+ 
		\|	  F_n \|_{L^2}^2
		-
		2 \< F , F_n \>  \ . 
	\end{equation} 
	For the  first term  in the right hand side  of \eqref{eq 1}, 
	we may use $F= F^2$
	and calculate $\|	F \|_{L^2}= 1$. 
	For the  second term, we use    the upper  bound 
	$\|	 F_n 	\|_{L^2} 
	\  \leq  \ 
	\| F_n	\|_{L^\infty}^{1/2}
	\|	F_n  \|_{L^1}^{1/2} $.
	Hence, $\limsup_{n\rightarrow \infty} \|	F_n\|_{L^2} \leq 1$. 
	For the third term, we note
	that   since the limiting function $F  $ is in $L^1$ 
	we may
	use the test function  $\phi  =  F$
	and  conclude that 
	\begin{equation}
		\lim_{N \rightarrow \infty} \< F , F_n \> 
		= 
		\<  F ,F  \> = 1 \ . 
	\end{equation}
	Thus, we put everything together to find that 
	\begin{equation}
		\limsup_{N\rightarrow \infty }
		\|	 F_n -F 	\|_{L^2}^2
		\leq 
		2      - 2 \lim_{N \rightarrow \infty} 	  	    \< F_n  , F\> 
		=  0 \  .
	\end{equation}
	This finishes the proof of the $p=2$ case. 
	
	\vspace{2mm}
	
	Let us prove the $p=1$ case. 
	In view of the identity 		$F (1 - F )=0$
	and the non-negativity of $F_n$
	and $1-F$
	we may consider the following decomposition 
	\begin{equation}
		\label{eq 2}
		\|	 F_n - F\|_{L^1}
		\	=  \ 
		\|	  F ( F - F_n  ) \|_{L^1}
		+ 
		\|	(1 - F  )   F_n  	\|_{L^1}
		\	\leq  \ 
		\|	F	\|_{L^2} \|	 F - F_n	\|_{L^2}
		+ \< 1 -F , F_n \> \ . 
	\end{equation}
	The first term in the right hand side  of \eqref{eq 2} is controlled by the difference in $L^2$ norm, analyzed above. 
	For the second term,  we use again 
	weak convergence  with the test function $\phi= 1 -F \in L^\infty $. 
	It suffices now to take the $n \rightarrow \infty$ limit, which finishes the $p=1$ case.

	\vspace{2mm} 
	
	Let us prove the $p\in(1,\infty)$ case. 
	For all such $p$, we use the fact that 
	$\|	 F_n -F  	\|_{L^\infty}\leq 2 $
	to find that 
	\begin{equation}
		\|	 F_n - F	\|_{L^p }^p 
		= \int 	| F_n - F |^p  \d \mu 
		=  
		\int 	|F_n  -F|  \ | F_n -F |^{p-1 } \d \mu 
		\leq 
		2^{p-1 }
		\|	 F - F_n	\|_{L^1} \ . 
	\end{equation}
	Because of the $p=1$ case, the right hand side now vanishes in the limit
	$n \rightarrow \infty$. This finishes the proof of the lemma. 
\end{proof}

\vspace{3mm}

Let us now turn to the proof of our main result, Theorem \ref{theorem 1}. 
The proof of its first part is a direct application of Lemma \ref{lemma abstract}. The proof of the second part
analyzes the difference between $f_N$ and $m_N$.

\begin{proof}[Proof of Theorem \ref{theorem 1}]
	Let $m_N$,  $f_N$  and $f$ be as in the statement of the theorem.
	
	\vspace{1mm}

	(1)  It suffices to use Lemma \ref{lemma abstract}
	with $X = \R^{2d}$, $\mathscr{F}$ the Borel sets and $\mu$ the $2d$-dimensional Lebesgue measure. 
	We consider the functions
	$F_N (z ) = 	m_N (  (2\pi)^{ 1/2}	z 	)	 	$
	and
	$F(z) = f ( (2\pi)^{ 1/2 }  z 	)$. 
	In particular, \eqref{cond Fn} and \eqref{cond F}
	are readily verified after a change of variables thanks to 
	\eqref{husimi conds}  and \eqref{definition f in rho}. 
	The weak convergence
	follows from  \eqref{convergence minimizers}
	and
	\eqref{convergence minimizers husimi}, 
	proven in \cite{Fournais2018}.
	
	\vspace{1mm}

	(2) 
	Let $|g|_{s,q}$
	be the norm introduced in \eqref{sobolev norm}.
	First, we prove the $q=\infty$ case. 
	Second, we prove the $q=2$ case. 
	Every other $ q \in (2, \infty)$ then follows by interpolation. 
	
	\vspace{2mm}
	
	\noindent \textit{The $q=\infty$ case}. 
	Let $s < 0$. 
	We use the triangle inequality  
	together with the elementary bound
	$  | \cdot |_{s,\infty} \leq \|	 \cdot 	\|_{L^1}$ to find that 
	\begin{align}
		| f_N - f  |_{s,\infty}
		& 
		\ \leq  \ 
		| f_N  - m_N|_{s,\infty}
		\ +  \ 
		\|	m_N - f \|_{L^1}
		\label{rhs0}
	\end{align}
	The second term in  \eqref{rhs0}  converges to zero, in view of part (1) of the Theorem. 
	Thus, it suffices to analyze   the first term in  \eqref{rhs0}.
	In particular, 
	we consider  only the case $| s|  \in ( 0,1]$, 
	since otherwise 
	we can use 
	$ | 	f _N - m_N 	|_{s,\infty} \leq 
	| 	f _N - m_N 	|_{1,\infty} 
	$
	in \eqref{rhs0}. 
	To this end,  we look at  its Fourier transform
	at $\zeta \in \R^{2d}$. 
	Namely, 
	let 
	$|g|_{C^{0, \alpha } }
	=
	\sup_{\zeta_1,\zeta_2 \in \R^{2d}} 
	| g(\zeta_1)  - g (\zeta_2) |	\, |\zeta_1-\zeta_2|^{-\alpha  }
	$
	the H\"older norm of a function $g : \R^{2d} \rightarrow \C$. Then,  we find
	\begin{align}
		|   (   \hat f_N - \hat m_N )(\zeta )| 
= 
		| \hat f_N   (\zeta) |  \, 
		| 1  -  \hat  {\mathscr G}_1 ( \hbar \zeta     )  | 
\leq 
		| \hat f_N (\zeta)   |    \, 
		| 		\hat {\mathscr G_1} |_{  C^{0, |s | }}
		\hbar^{|s|}    |\zeta|^{|s| } 
 \leq 
		| 		\hat {\mathscr G_1} |_{ C^{ 0,| s| }}  \, 
		\hbar^{ |s| }    |\zeta|^{|s| } 
		\label{rhs 1}
	\end{align}
	where we used $\hat{\mathscr G_1	}(0) =1$, 
	together with the  uniform bound 
	$\|  \hat f_N	\|_{L^\infty } \leq 1$ (see \eqref{bounds L inf}).
	The last inequality now implies that  
	\begin{equation}
		|  f_N  - m_N |_{s,\infty }
		\ 		 =   \ 
		\sup_{ \zeta  \in \R^{2d} }	
		\<  \zeta   \>^{s} 	
		|  ( \hat f_N - \hat m_N )( \zeta  )| 
		\ \leq \  
		| 		\hat {\mathscr G_1} |_{		C^{0,s}} 
		\hbar^{|s|	}  . 
		\label{rhs 2}
	\end{equation}
	This finishes the proof thanks to the fact  that 
	$	| 		\hat {\mathscr G_1} |_{		C^{0,| s| }}  < \infty$ for $|s| \in ( 0,1]. $
	
	\vspace{2mm}
	
	\noindent \textit{The $q=2$ case}. 
	The proof is identical to the proof of (2),  with the following changes. 
	First,  one replaces $\|	m_N - f\|_{L^1}$ with $\|	m_N  - f \|_{L^2}$
	and uses part (1) for $p=2$.
	Second, one replaces the uniform bound $\| \widehat{ f} _N \|_{L^\infty} \leq 1 $
	with 
	the  alternative uniform $L^2$-bound $\| \widehat{f}_N 	\|_{L^2} \leq (2\pi)^{\frac{d}{2}}$. 
	The  latter  bounds  follows 
	from the  fact that  (up to scaling)
	the Wigner transformation is a unitary map between $L_{x,x'}^2$ and $L_{x,p}^2$.
	More precisely, in our setting we obtain 
	\begin{align} 
		\label{L^2}
		\|	f_N	\|_{L^2} 
		\ = \ 
		(2\pi \hbar)^{\frac{d}{2}}
		\|	 \gamma_{N} 	\|_{L^2}	   
		\  		=  \ 
		(2\pi \hbar)^{\frac{d}{2}}
		\|	 \gamma_{N} 	\|_{HS}	  
		\  		 \leq  \ 
		(2\pi \hbar)^{\frac{d}{2}}
		\|	\gamma_N	\|_{\tr}^{\frac{1}{2}}
		\  		=  \ 
		(2\pi  )^{\frac{d}{2}} \ . 
	\end{align}
	In the last two inequalities we used the fact that 
	$ \|	\gamma\|_{HS} = (  \tr \gamma^*\gamma )^{\frac{1}{2}}
	\leq 
	\|	\gamma\|^{\frac{1}{2}}	_{B(L^2)} \| 	\gamma \|_{\tr}^{\frac{1}{2}}$
	for arbitrary trace-class operators on $L^2$, 
	the fact that for fermions $ 0 \leq \gamma \leq 1 $, 
	and the scaling $ \tr \gamma = N = \hbar^{- d}$. 
\end{proof}

\vspace{2mm}

\begin{proof}[Proof of Corollary \ref{corollary}]
	Let $p$, $q$ and $r$ be as in the statement of the corollary.
	Then, we use the triangle inequality to find that 
	\begin{equation}
		\|	f_N   -f 	\|_{L^p} \leq \|	f_N	-	m_N	\|_{L^p}	+ \|	 m_N - f 	\|_{L^p}  \ . 
	\end{equation}
	Thanks to Theorem \ref{theorem 1}, $\lim_{ N \rightarrow \infty }	\| m_N -f	\|_{L^p} = 0 $. Thus,     we only look at the first term.
	We write for $z = (x,p) \in \R^{2d}$
	\begin{equation}
		f_N (z ) - m_N(z)
		= 
		\int_{	\R^{2d}	}
		\big(    
		f_N(z  ) - f_N( z  - \hbar  z' )
		\big) 
		\mathscr{G}_1 (z') \d z'  \ . 
	\end{equation}
	Hence, Minkowski's and H\"older's inequality implies 
	\begin{align}
		\nonumber
		\|	f_N - m_N	\|_{L^p}
		& \ \leq  \ 
		\int_{\R^{2d}}
		\|	 f_N  -		f_N (\bullet + \hbar z' )			\|_{L^p}
		\mathscr{G}_1 (z' ) d z' 	\\
		\nonumber
		&  \  =   \ 
		\int_{\R^{2d}}
		\frac{\|	 f_N  -		f_N (\bullet + \hbar z' )			\|_{L^p}}{	|z'|^r	}
		\, |z'|^r
		\mathscr{G}_1 (z' ) d z'  	\\ 
		&  \ \leq  \ 
		\bigg\|
		\frac{\|	 f_N  -		f_N (\bullet + \hbar z' )			\|_{L^p}}{	|z'|^r	}
		\bigg\|_{L^q(d z ' )}
		\|	 |z '|^r \mathscr{G}_1 	\|_{L^{q'}(dz ' )}
	\end{align}
	where $1/q'=1-1/q$ is the dual exponent of $q$. 
	Note that 
	$\|	 |z '|^r \mathscr{G}_1 	\|_{L^{q'}(dz ' )}<\infty$
	because $\mathscr{G}_1$ is a Schwartz-class function. 
	In view of the assumption of the theorem,  it suffices to take 
	the limit $N \rightarrow \infty$. 
\end{proof}

\subsection{Proof of Theorem \ref{thm 2}}
\label{subsection 3.2}
Recall that
in  order to prove  that $f_N$ converges to $f$  in a Banach space $E$, 
it suffices to show that 
from  all subsequences $(f_{N_k})_{k \geq 1 }$
we can extract a further subsequence 
converging to $f$ in $E$. 
We shall be using this fact in the proof of Theorem \ref{thm 2} part (1) with 
$E= L^2  $ 
and Corollary \ref{corollary 2}
with $E = L^1  $, but make no explicit reference to it in 
order to simplify the exposition.

\vspace{2mm}

\begin{proof}[Proof of Theorem \ref{thm 2}]
	Let $\gamma_N$,  $f_N$ and $f$ be as in the statement of the theorem. 
	
	\vspace{2mm}
	(1) 
	When $\gamma_N$ is a rank-$N$ orthogonal projection, 
	its Hilbert-Schmidt norm  can be calculated to be
	$\|	  \gamma_N	\|_{HS} = (\tr \gamma_N^* \gamma_N)^{\frac{1}{2}} 
	= 
	(\tr \gamma_N)^{\frac{1}{2}}
	= N^{\frac{1}{2}}. 		$
	Thus, the same reasoning that led to  \eqref{L^2}  now yields 
	\begin{equation}
		\|	f_N	\|_{L^2} 
		\  		=  \ 
		(2\pi  )^{\frac{d}{2}} \ . 
	\end{equation}
	Thus,   we may assume that  $f_N$ converges weakly in $L^2$.
	Of course,  in view of \eqref{convergence minimizers}, this limit is given by $f $. 
	Further, $f = f^2$ implies that  $\|	f	\|_{L^2}  = (2\pi)^{d/2}$.
	This shows that 
	$\lim_{N\rightarrow \infty} \|f_N	\|_{L^2}  =  \|	f\|_{L^2}$.
	Since $L^2$ is a Uniformly Convex Space, 
	we may conclude that convergence is strong in $L^2$, 
	and this   finishes the   proof of the first part of the theorem. 
	
	\vspace{3mm}
	
	(2) 
	Let $m>0$ and $ p \in [1,2] \cap (	 \frac{2}{1+ 2m/d}	 ,2 ] $ be as in the statement of the theorem, and let us denote by 
	\begin{equation}
		K_m \equiv \sup_{N \geq 1 }
		\|    \< z  \>^m f_N		\|_{L^2} < \infty \ . 
	\end{equation}
	Let $R>0$
	and denote   $B_R \equiv \{z \in \R^{2d} : |z| \leq R \}$
	and $B_R^c \equiv \R^{2d} / B_R$.
	We let $ \chi_R \equiv \mathds{1}_{B_R}$
	and $\chi^c_R \equiv 1 - \chi_R$.
	
	\vspace{1mm}
	
	The proof starts with the  decomposition  
	\begin{equation}
		\label{eq thm 2 1}
		\|	f - f_N	\|_{L^p}
		\leq 
		\|   \chi_R (f - f_N)		\|_{L^p}
		+ 
		\|   \chi_R^c (f - f_N)		\|_{L^p} \ . 
	\end{equation}
	We estimate each term separately. The first term can be estimated 
	in terms of the $L^2$ norm, using 
	H\"older's inequality 
	\begin{equation}
		\label{eq thm 2 2}
		\|   \chi_R (f - f_N)		\|_{L^p} 
		\leq 
		| B_R|^{\frac{2-p}{2p}}\|	f - f_N\|_{L^2}\ . 
	\end{equation}
	Note that thanks to the first part of the theorem, the right hand side vanishes in the limit $N \rightarrow \infty$.  	For the second term of \eqref{eq thm 2 1},  we use the triangle inequality 
	\begin{equation}
		\label{eq thm 2 3}
		\|   \chi_R^c (f - f_N)		\|_{L^p}  
		\leq 
		\|   \chi_R^c  f 		\|_{L^p}   + 
		\|   \chi_R^c  f_N	\|_{L^p}  \ . 
	\end{equation}
	Note that $f \in L^p$ and so the first term converges to zero as $R \rightarrow \infty$.
	In order to control the second one, we use the moments of  $f_N$.
	Namely, H\"older's inequality gives 
	\begin{align}
		\label{eq thm 2 4}
		\|	 \chi_R^c f_N	\|_{L^p}
		\  =		 \ 
		\|	  \chi_R^c 	\< z\>^{-m}	\< z\>^m					f_N	\|_{L^p} 	 
		\ \leq  \ 
		\|	  \chi_R^c 	\< z\>^{-m}	\|_{L^{r'}}
		\|	\< z\>^m					f_N	\|_{L^r}
	\end{align}
	where $1/p = 1/r + 1/r'$. 
	We choose $r=2 $ and $r' = 2p / (2-p)$.
	In particular, 
	we observe that $\< z \>^{-m} \in L^{r'}$.
	Indeed, this is equivalent to  $ m r ' > d$, which  holds for 
	$p > 2 / (1 + \frac{2m}{d})$. 
	
	\vspace{1mm}
	To finish the proof, we take the $\limsup_{N \rightarrow \infty}$
	and use the inequalities 
	\eqref{eq thm 2 2},
	\eqref{eq thm 2 3}
	and 
	\eqref{eq thm 2 4}
	to find that 
	\begin{equation}
		\limsup_{N \rightarrow \infty} \|	f - f_N	\|_{L^p}
		\leq \|	  \chi_R^c f	\|_{L^p}
		+ K_m
		\|	  \chi_R^c 	\< z\>^{-m}	\|_{L^{r'}} \ . 
	\end{equation}
	It suffices to take the limit $R \rightarrow \infty$.
	This finishes the proof of the theorem.    
\end{proof}

\vspace{2mm} 

\begin{proof}[Proof of Corollary \ref{corollary 2}]
	For simplicity 	we only present the proof for $ p=1 $,   the other cases being analogous with only minor modifications.
	Since one implication is obvious,   we only prove the second one.

	\vspace{2mm}
	
We start  by using  the fact  $f=f^2$ 
	through the formula $f(1 - f)=0$, 
and estimate the $L^1$ norm difference as follows 
	\begin{align}
		\label{L1 step 1}
		\|	f - f_N	\|_{L^1}
		= 
		\int  f | f - f_N| 
		+ 
		\int  (1-f) | f - f_N| 
		= 
		\int f | f - f_N| 
		+ 
		\int  (1-f) | f_N|  \ . 
	\end{align}
	The first term in \eqref{L1 step 1} can be estimated
	in view of the Cauchy-Schwarz inequality, and the $p=2$ case established in Theorem \ref{thm 2}. Namely 
	\begin{equation}
		\int f | f - f_N| 
		\leq
		\|	f 	\|_{L^2} 
		\|	f - f_N	\|_{L^2}  = (2\pi)^{\frac{d}{2}}	  \|	f - f_N	\|_{L^2}  \ . 
	\end{equation}
	The second term in \eqref{L1 step 1} is more involved. 
	First, because 
	of $L^2$ convergence, 
	we may assume that 
	convergence holds pointwise almost everywhere, up to
	extraction of a subsequence (which we do not display explicitly).  
	Hence, Fatou's Lemma implies  that 
	\begin{equation}
		\label{eq3}
		(2\pi)^d 
		\ = \ 
		\int f |f|  
		\ 			\leq \ 
		\liminf_{N \rightarrow \infty}
		\int  f |f_{N  } | \ . 
	\end{equation}
	On the other hand,    $ \|	f_N\|_{L^2} = \| f \|_{L^2 }= (2\pi)^{\frac{d}{2}}$  combined with 
	the Cauchy-Schwarz inequality implies 
	\begin{equation}
		\label{eq4}
		\limsup_{N \rightarrow \infty}
		\int  f |f_N|  
		\ \leq  \ 
		\| f \|_{L^2} 
		\limsup_{N \rightarrow \infty}
		\|	f_N	\|_{L^2} 
		=			(2\pi)^d
		\ . 
	\end{equation}
	Hence,  \eqref{eq3} and \eqref{eq4} imply that 
	$\lim_{N \rightarrow \infty } \int  f |f_N| =			(2\pi)^d$.  
	Thus, going back to \eqref{L1 step 1} we see that 
	\begin{equation}
		\limsup_{N \rightarrow \infty} \|	 f - f_N	\|_{L^1} 
		\ 	\leq  \ 
		\limsup_{N\rightarrow \infty }  \| f_N	\|_{L^1} -
		(2\pi)^d  \ . 
	\end{equation}
Thus, thanks 
to the   assumption in the statement of the Corollary,  we find
$\lim_{ N \rightarrow \infty } \|	f_N\|_{L^1} = (2\pi)^d$. 
Finally, we note that the  limit $f$ is
independent of the chosen subsequence. 
Hence, $L^1$-convergence holds for the entire sequence (see also the comment at the beginning of Subsection \ref{subsection 3.2}). 
\end{proof}

\section{Application: the harmonic trap}
\label{section harmonic}

Throughout this section, 
we   consider   the external potential $U(x) = x^2$. 
Using the same notation as in the previous section, 
we denote by $f_N$ the sequence of Wigner functions
of the approximate ground state $\Psi_N$,
taken to be a Slater determinant in the context of Theorem \ref{thm 2}. 
We denote by   $f = f_{TF}$ the classical state
determined by $\rho_{TF}$, the minimizer of $e_{TF}$. 

\vspace{1mm}

In the following lemma, we prove that the condition on the moments \eqref{moments}
is verified for $m =1$. 
Note this is the only point in this article
where we make direct use of the 
conditions $A = 0 $ and $\hat V \geq 0$. 
We believe further analysis would
show these can be relaxed to those considered originally in \cite{Fournais2018}.

\begin{lemma}
	\label{lemma moments}
Assume that   $\Psi_N$ is a Slater determinant for all $N\geq1.$
Then, 	for $U(x) = x^2$
there holds  
	\begin{equation}
		\sup_{N \geq 1} \|    (|x| + |p|)	f_N	\|_{L^2} < \infty \ . 
	\end{equation}
\end{lemma}

\begin{remark}
As it will be clear from the proof, one only needs a lower bound
$U(x) \geq C x^2$   for the external trap, 
and the same result is  true. 
For simplicity of the exposition we choose  the harmonic  trap $U (x) = x^2$. 
\end{remark}

\begin{remark}
	The estimate \eqref{eq moment 2} 
	contained in  the  proof  of Lemma \ref{lemma moments}   below 
	is borrowed  from  the proof of \cite[Theorem 3]{Lafleche23}.
	Here, 
	the author considers 
	the convergence 
	of  the  Wigner transform of the   following  sequence of 
	one-particle density matrices  
	\begin{equation}
		\label{rho1}
		\rho_\hbar = \1_{( - \infty, 0 ]	} \big(
		\hat p ^2 + W (\hat x )
		\big) 
	\end{equation}
	where $ W : \R^d \rightarrow \R$ is a nice         potential 
	chosen so 
	  that the spectrum of $p^2 + W(x)$ is discrete  in $( - \infty, 0]$.
	These states  were first studied in \cite{Soren19}
	for $d \geq 2 $
	where it is  
	shown that  they  satisfy 
	optimal semi-classical commutator estimates. 
\end{remark}

\begin{corollary}\label{corollary 3}
	Under the same assumptions as in Theorem \ref{thm 2}
	and Lemma \ref{lemma moments}, there holds
	\begin{equation}
	f_N \rightarrow f \qquad 
\t{strongly in }  L^p(\R^{2d}) \quad \t{for } 
p \in [1,2] \cap \Big(  \frac{2}{1 + 2/d} , 2  		\Big]  \ . 
	\end{equation}
\end{corollary}

\begin{remark}[One-dimensional harmonic oscillator]
	Let $d=1$ and consider the harmonic oscillator
	$h = - \hbar^2 \d^2/ \d x^2 + x^2$
	on $L^2(\R)$, 
	with spectrum $\sigma (h) = \cup_{n=0}^\infty \{ E_n\}$
	where $E_n = \hbar ( n + 1/2)$.
Each eigenvalue $E_n$
is non-degenerate
and corresponds to the eigenfunction $\psi_n$
being the $n$-th Hermite function.
Since $(\psi_n)_n$
is an orthonormal basis for $L^2(\R)$, 
it is possible to show that the 
ground state 
of 
$H_N = \sum_{ i = 1 }^N - \hbar^2 \d / \d x_i^2 + x_i^2$
on $\bigwedge_{i=1}^N L^2(\R)$
corresponds to the Slater determinant 
\begin{equation}
	\Psi_{N,0} =  \psi_0 \wedge \cdots \wedge \psi_{N-1}   \ , 
\end{equation}
with ground state energy $\sum_{ n=0 }^{N-1} E_n	$. 
Letting $f_{N,0}$ be the Wigner transform
of the ground state $\Psi_{N,0}$, 
Corollary \ref{corollary 3}
  implies   
\begin{equation}
	f_{N,0} \rightarrow  f_{\rho_0}	\qquad		\t{strongly in } L^1(\R^2 ) \ . 
\end{equation}
where we denote by $\rho_0 \in L^1(\R)$
the minimizer of the   Thomas-Fermi functional. 
\end{remark}

\begin{proof}[Proof of Lemma \ref{lemma moments}]
Let us  denote by $\hat x$ and $\hat  p   = -  i \hbar \nabla_x$ the standard position and momentum observables
in $L^2(\R^d)$.  We start the proof	
with the identity 
\begin{align}
		x f_N    =  \, (1/2) \,  W^\hbar[ \hat x \gamma_N  +  \gamma_N \hat x	]  , 
	\end{align}
which follows
from the decomposition 
 $x = 1/2(x + y/2)  +   1/2(x - y/2)$, the definition \eqref{wigner def}
 and the relations
 $ ( \hat x \gamma) (x,y) =  x \gamma(x,y) $
 and
 $
 (\gamma \hat x) (x,y) = 
 y \gamma(x,y). 
 $
Thus, $\|	W^\hbar[\gamma] \|_{L^2 } = (2\pi\hbar)^{d/2} \|	 \gamma\|_{L^2}$
implies 
\begin{align}
\label{eq moment 1}
	\|   x  f_N 	\|_{L^2}
	 \leq 
 (1/2)(2\pi \hbar)^{d/2}
	 \Big(
	 \|	 \hat x	 \gamma_N \|_{L^2}
	 + 
	 	 \|	 	 \gamma_N  \hat x \|_{L^2}
	 \Big) \  . 
\end{align}
We now calculate the right hand side as follows 
\begin{equation}
	\| \hat x \gamma_N	\|_{L^2}^2
	 = \|	 \hat x \gamma_N	\|_{HS}^2
	  = \tr  ( 
	  \hat x \gamma_N
)^*  (\hat x \gamma_N)
	  = 
	  \tr \Big(
	  \gamma_N \hat x ^2 \gamma_N  
	  \Big)
	  = 
	    \tr \Big(
 \hat x ^2 \gamma_N  
	  \Big) \   ,
\end{equation}
wher  we used cyclicity of the trace, and 
 $\gamma_N = \gamma_N^* =  \gamma_N^2 $. 
 The same identity holds for the second term of \eqref{eq moment 1}.  Thus, 
 \begin{equation}
 	\|	  x f_N	\|_{L^2}^2 
 	\leq 
  (2\pi	\hbar)^d \tr \big(
 	\hat x ^2 \gamma_N 
 	\big) \ .
 \end{equation}
In a similar fashion, the same argument can be repeated  for the momentum variable 
and can be combined with   the last estimate. 
Thus, the triangle inequality and the scaling $\hbar^d =N^{-1}$ gives 
\begin{equation}
\label{eq moment 2}
	 	\|	  ( |x|  + |p|)  f_N	\|_{L^2}^2 
	\leq 
\frac{C}{N} \tr \big(
	( \hat x ^2 +  \hat p ^2)  \gamma_N 
	\big) \ , 
\end{equation}
for a constant $C>0$. 
It suffices to show that the right hand side of \eqref{eq moment 2}
is bounded uniformly in $N$.
To this end, we use two facts:  that $\gamma_N$ corresponds
to an   approximate ground state; 
and that the two-body interaction is
dominated by the kinetic energy. 
We do this in the following two steps. 

\vspace{1mm}

\textit{Step 1.} Let $\Psi_N$ be
the approximate ground state 
under consideration, 
which we assume is a Slater determinant. 
Then, a standard calculation shows that 
\begin{align}
\label{eq1}
\langle  &  \Psi_N ,  H_N \Psi_N \rangle_{\mathfrak{h}_N} 	 \\ 
&   = 
   \tr  \big(
   \hat p ^2 + \hat x^2 
\big) \gamma_N 
+ 
N 
 \int_{\R^{2d}} 
\rho_N (x) V(x-y) \rho_N(y) \d x \d y 
  - 
  \frac{1}{N }
  \int_{	\R^{2d}	} V(x-  y ) |\gamma (x,y)|^2 \d x \d y  \ . 
  \nonumber 
\end{align}
where we denote $\rho_N(x) \equiv \frac{1}{N} \gamma_N (x,x)$. 
Recall that in  Assumption \ref{assumption 1}
we consider $\hat V \geq 0 $. Hence,  the second term on the right hand side above is non-negative, in view of the following representation 
\begin{equation}
\label{eq2}
	\int_{	\R^{2d}	} \rho (x) V(x-y) \rho(y) \d x \d y 
	 = 
	 \int_{	\R^{d}	} \hat V(k) |	 \hat \rho (k)|^2 \d k \geq 0 . 
\end{equation}
Therefore, thanks to   \eqref{eq1}, \eqref{eq2} and  \eqref{asymptotics}   
 we conclude that for $N \geq N_0$ large enough 
\begin{align}
	   \tr  \big(
	\hat p ^2 + \hat x^2 
	\big) \gamma_N 
& 	 \leq C N e_{TF}
+ 	 \frac{1}{N}	 	 \int_{	\R^{2d}	} V(x-y)  |\gamma_N (x,y)|^2 \d x \d y \  . 
\label{eq moment 3}
\end{align}

\vspace{2mm}

\textit{Step 2}. 
Here, we borrow  the following estimate 
from  the proof of   \cite[Proposition 3.1]{Fournais2018}
\begin{equation}
\label{eq moment 4}
\bigg| 
	 \frac{1}{N}
	\int_{\R^{d} \times \R^d}
	V(x-y)  |\gamma (x,y)|^2 \d x \d y 
\bigg| 
	\leq    \ve_N 
 \tr  \big( 
\hat p ^2  + \1   \big)  \gamma_N \  
\end{equation}
for a sequence of positive numbers $\ve_N \rightarrow 0$
as $ N \rightarrow \infty$. 
Consequently, for $N \geq N_0$ large enough we may assume $\ve_N\leq 1/2$. 
In particular, it follows from \eqref{eq moment 3},  \eqref{eq moment 4} and $\tr \gamma_N= N$
that  for a constant $C>0$
\begin{equation}
	\tr(\hat p ^2 + \hat x ^2) \gamma_N 
	\leq C N \big(
	e_{TF}  + 1 	\big) \ . 
\end{equation}
The proof is  finished once we combine the last inequality
with \eqref{eq moment 2}. 
\end{proof}

\section{Results on $k$-particle distribution functions}
\label{section higher}
 In this section we address the problem of convergence
 of higher-order distribution functions, as originally considered by the authors in \cite{Fournais2018}.

\subsection{Definitions}
Let $  2 \leq k <  N   $ 
and denote by  $\Psi_N	\in\mathfrak{ h }_N$   the approximate ground state.
We will work with the following definitions.

 \vspace{2mm}

(1) \textit{Wigner functions.}
Following  \cite[Eq. (1.18)]{Fournais2018} we define 	 
the 
$k$-particle Wigner function as follows 
\begin{equation}
	\label{def wigner k}
	f_N^{(k)}( X_k, P_k)
	\equiv 
	\int_{\R^{2dk} \times \R^{d(N-k)}}
		\Psi_N \Big(  X_k + \frac{\hbar }{2}Y_k  , X_{N-k} \Big) 
	\overline{ \Psi_N }{ \Big(  X_k + \frac{\hbar }{2}Y_k		  , X_{N-k}\Big) } 
	e^{	   -  i Y_k \cdot P_k }	
	\d X_{N - k } \d Y_k 
\end{equation} 
where  $(X_k, P_k) \in \R^{2dk}$, 
and we write  $X_{N-k} = (x_{k+1} , \ldots, x_N ) \in \R^{d (N -k )}$
for   the  integration variable. 
In particular, 
the normalization is chosen so that 
\begin{equation}
	\frac{1}{(2\pi)^{dk}}
	\int_{	\R^{2dk}	}
	f_N^{(k)} (X_k , P_k) \d X_k \d P_k =1 \   \ . 
\end{equation}

\vspace{2mm}

(2) \textit{Husimi measures.}
In this article, we define 
the $k$-particle Husimi measure 
as the convolution  
\begin{equation}
\label{husimi k}
	m_N^{(k)}  \equiv 
	\frac{N \cdots (N  - k + 1 )}{N^k }
	f_N^{(k)}*  \mathscr{G}_\hbar^{\otimes k }
\end{equation}
where 
$\mathscr{G}_\hbar^{\otimes k } = \mathscr{G}_\hbar \otimes \cdots \otimes \mathscr{G}_\hbar $
is the $k$-fold tensor product of the one-particle 
Gaussian mollifier
$\mathscr{G}_\hbar (z)= \hbar^{-d} \mathscr{G}_1 ( \hbar^{- 1/2} z )$.  
As explained by the authors in \cite[Section 2.4]{Fournais2018}, 
the convolution \eqref{husimi k}
agrees with their original definition.
In particular, 
we have \cite[Lemma 2.2]{Fournais2018}
\begin{equation}
	\label{k husimi bounds}
	0 \leq m_N^{(k)} \leq 1
	\qquad \t{and} \qquad 
	\frac{1}{(2\pi)^{dk}}
	\|	 m_N^{(k)}\|_{L^1}= 
	\frac{N  \cdots (N - k +1 )	}{N^k} \ . 
\end{equation}

\vspace{2mm}

(3) \textit{Reduced densities.}
We define $\gamma_N^{(k)}$, the  $k$-particle reduced density matrix, 
as  the trace-class
operator
on
$L^2(\R^{dk})$ with   kernel 
\begin{equation}
\label{def density k}
	\gamma_N^{(k)} (X_k , X_k ' )
	\equiv 
 \frac{N!}{(N-k)!}
	\int_{\R^{d(N-k)}}
	\Psi_N(  X_k , X_{N-k})
	\overline{ \Psi_N( X_k'  , X_{N-k})} 
	\ 		\d X_{N-k} 
\end{equation}
where 
$X_k = (x_1, \ldots, x_k), \, X_k' = (x_1' ,\ldots, x_k') \in \R^{dk}$
and  we denote $X_{N-k}  = (x_{k+1} , \ldots, x_{N}) $. 
In particular, the normalization is chosen so that
\begin{equation}
	\label{gamma bounds} 
	\tr \gamma_N^{(k)} 
	=
	 \frac{N!}{(N-k)!}  \ . 
\end{equation} 
Note the definition is not the same 
as   \cite{Fournais2018},  
but we take the one from  \cite{LiebSeiringer} since we shall borrow a calculation
in the proof of Lemma \ref{lemma L2} below.

\subsection{Statements} 
	In \cite{Fournais2018}
the authors prove  that the following convergence statement holds 
\begin{equation}
\label{husimi convg}
 \< \vp , 	 m_N^{(k)} \> \rightarrow \< \vp ,  f^{\otimes k }		\>
\end{equation}  
for all $\vp \in L^1 (\R^{2dk}) + L^\infty(\R^{2dk})$, 
where   $f = f_{\rho_{TF}}$ is  defined
in terms of  the unique minimizer of the Thomas-Fermi functional \eqref{thomas fermi}. 
In particular, one may then ask if 
the results contained in Theorem \ref{theorem 1} and \ref{thm 2}
extend to higher-order distribution functions.

\vspace{1mm}

We address this question in the next Theorem.
While some convergence results are easily translated 
into higher-order distribution functions, it turns out that some are not. 
Most notably, we prove that  for Slater determinants 
the strong $L^2$ convergence 
 $f_N^{(k)} \rightarrow f^{\otimes k }$
holds \textit{only} for $k=1$, 
although convergence holds in $H^{s}$ for all $s<0. $

%
%

\begin{theorem}[$k$-particle functions]
	\label{thm 3}
Let $ k \geq2$. Then, 	the following statements are true.  
	\begin{enumerate}[leftmargin=1cm]
		\item 
		The sequence $m^{(k)}_N$ converges  to $f^{\otimes k }$ strongly in $L^p(\R^{2dk})$ for all 
		$ p \in [1,\infty)$.

		\item 
		The sequence $f_N^{(k)}$ converges to $f^{\otimes k}$ strongly in $H_{s,\infty}(\R^{2dk})$ 
for all $s< 0$. 
	\end{enumerate}	
	Additionally,   if  $\Psi_N$ is a Slater determinant for all $N\geq1$, 
	the following is true. 
	\begin{enumerate}[leftmargin=1cm]
\item[(3)]
The sequence $f_N^{(k)}$ converges to $f^{\otimes k}$ strongly in $H_{s,q }(\R^{2dk})$ 
for all $s< 0$ and $q \in [2,\infty]$. 

\item[(4)]  The sequence $f_N^{(k)}$ 
\underline{does not} converge to $f^{\otimes k }$
strongly in $L^2 (\R^{2dk})$. 
In fact, there holds 
\begin{equation}
\liminf_{N\rightarrow \infty}		\|	f_N^{(k)} - f^{\otimes k }	\|_{L^2}
\geq 
(2\pi )^{dk/2}\big(
\sqrt{ k!	} -1 
\big). 	
\end{equation}
	\end{enumerate}
\end{theorem}

 \begin{remark}[Lack of uniform $L^2$ bounds]
Note that   
in item (2) we  are not 
able to obtain   convergence in $ H^{	 s }(\R^{2dk})$, 
as opposed to the $k=1$ case in Theorem \ref{theorem 1}. 
This is a consequence of our    lack of control of the $L^2$ norms of   $f_N^{(k)}$
for $ k \geq 2$. 
Recall that   for  density matrices $\gamma_N^{(k)}$ for $ k \geq 2 $
one does not  have  operator bounds which are uniform in $N$. 
 	For  illustration, for $ k  =  2$ it is in general known that 
\begin{equation}
\label{op bound}
	0  \leq \gamma_N^{(2)} \leq (N-1)\1  \  , 
\end{equation}
see e.g 
\cite[Lemma 1]{Bach} or \cite[Theorem 8.12]{Solovej2014}. 
In contrast,  for the case  $k=1 $, there holds   $0 \leq \gamma_N^{(1) } \leq \1 $, 
which implies $\|	 f_N^{(1 )}	\| _{L^2} \leq (2\pi)^{d/2}$. 
On the other hand, 
the operator bound \eqref{op bound} 
only gives the estimate 
$\|	 \gamma_N^{ (2)}	\|_{HS} \leq
\|	 \gamma_N^{(2)}	\|^{	 \frac{1}{2} } 
\|	 \gamma_N^{ (2)}	\|^{ \frac{1}{2} }_{\tr}
\leq N^{3/2}, 
$
which is not sufficient to guarantee 
a uniform $L^2$ bound for $f_{N}^{(2)}$. 
Let us further make two comments.

\begin{enumerate}[leftmargin=*]
	\item 
In item (3), the uniform $L^2$ bounds
are    recovered for the Slater determinants, for which 
an explicit formula can be given for all $ k \geq 2$; 
see Lemma \ref{lemma L2}. 
This allows for convergence in the larger class of spaces. 

	\item 
	
		Upon completion of this work, 
		the author became 
	aware of the recent result of Christiansen \cite[Theorem 1]{Christiansen}\footnote{The author would like to thank the anonymous referee that pointed this out}, 
	where it is proven that $ \|\gamma_N^{(2)}	\|_{HS}  \leq \sqrt 5 N$. 
	The proof bypasses the operator norm  \eqref{op bound} and    holds 
	for \textit{any} normalized state $\Psi_N$. 
Consequently, $\sup_{N \geq 1 } \|	 f_N^{(2)}	\|_{L^2}<\infty$ 
using the unitarity of the Wigner transform (see also \eqref{f HS}).  
Hence, we can relax the assumption of $\Psi_N$
being a Slater determinant in 
 item (3) in Theorem \ref{thm 3}
for $k=2$. 
In particular, $f_N^{(2)} \rightarrow f \otimes f $
in $H^s(\R^{4d})$ for all $s<0$.

\end{enumerate}

 \end{remark}

The proof of items (3) and (4) of Theorem \ref{thm 3}
are based on the following calculation. 
While the result may have been known to experts, 
we record it here for completeness of the exposition.

 \begin{lemma}\label{lemma L2}
 	Assume that $\Psi_N$ is  a Slater determinant for all $N \geq 2$.
 	Then, for all  $ k \geq 2$ 
 	\begin{equation}
\lim_{N \rightarrow \infty } 		\|	f_N^{(k)}	\|_{L^2}
 =
   \Big(   (2\pi)^{dk } k! 	\Big)^{\frac{1}{2}}	\ . 
 	\end{equation}
 \end{lemma}
 
 \begin{remark}
The $L^2$ norm of $f_N^{(k)}$ 	
can be computed explicitly for finite $N$, see \eqref{f L2}. 
 \end{remark}

 \subsection{Proof of Theorem \ref{thm 3}}
First,  we shall give a proof of Theorem  \ref{thm 3}  which is based on Lemma \ref{lemma L2}. Then, we prove Lemma \ref{lemma L2}. 

\vspace{1mm}

Before we turn to the proof of Theorem \ref{thm 3}, let us introduce some notation.
Namely, it will be extremely convenient
to recast the Wigner function $f_N^{(k)}$
as the Wigner transform of
the reduced density matrix $\gamma_{ N}^{(k)}$. 
To this end, 
we define the  following  linear map, 
which we refer to as the Wigner transform of order $ k$: 
\begin{equation}
	W_k^\hbar[	 \gamma_k	](X_k , P_k)
	\equiv 
	\int_{\R^{dk}}
	\gamma_k 
	\Big( 
	X_k + \frac{1}{2} Y_k , 
	X_k - \frac{1}{2}Y_k 
	\Big) 
	\exp\Big(
	- \frac{i}{\hbar} Y_k\cdot P_k 
	\Big) \  \d Y_k
\end{equation}
where    $\gamma_k \in L^2(	 \R^{2dk}	)$
and
$(X_k, P_k) \in \R^{2dk}$. 
In particular,  the following is true. 
\begin{enumerate}[label=(\roman*)]
	\item  
	The map $W_k^\hbar$ is an $L^2$-isomorphism, and with the present normalization
	\begin{equation}
		\label{f 2}
		\|	 W_k^\hbar [\gamma_k]	\|_{L^2}	
		=
		(2\pi\hbar)^{\frac{dk}{2 }	} 
		\|	\gamma_k\|_{L^2} \ , 	 \qquad \gamma_k \in L^2(\R^{2dk}) \ . 
	\end{equation}
	See e.g \cite[Proposition 13]{Combescure2012}. 
	
	\item  Starting from \eqref{def wigner k}, 
	we   use the definition \eqref{def density k} of $\gamma_{ N}^{(k)}$, 
	change variables $Y_k \mapsto \hbar^{-1} Y_k $ 
	and use the scaling $ \hbar^{-dk } = N^k$
	to find that 
	
	\begin{equation}
		\label{f 1}
		f_N^{(k)	}=	N^k \ \frac{(N-k)!}{N!}	\ W_k^\hbar 	  [\gamma_N^{(k)}]	  \ . 
	\end{equation} 	  
	Note that $\lim_{ N \rightarrow \infty } N^k \ \frac{(N-k)!}{N!}	 =1 $. 
\end{enumerate}

 \vspace{1mm}
 
 \begin{proof}[Proof of Theorem \ref{thm 3}]
 	In what follows we fix $k\geq 2 $
 	and assume $N >  k $.  
 	
 	\vspace{2mm}
 	
(1) 
The proof relies on Lemma \ref{lemma abstract}.
Here, we take $F_n( Z_k) = m_N^{(k)}( (2\pi)^{\frac{1}{2} } 	 Z_k	)$
and similarly
$F (Z_k) =   ( f^{\otimes k }  )( (2\pi)^{\frac{1}{2}}	Z_k	)	$
for $Z_k \in \R^{2dk}$. 
The bounds \eqref{cond Fn} and \eqref{cond F} 
are straightforward to verify
and follow from \eqref{k husimi bounds}
and $\int f  =  (2\pi)^{d/2}$  and $f =f^2 $, 
respectively. 
The weak convergence \eqref{husimi convg}
has been proven in \cite{Fournais2018}. 
%
 	
 	\vspace{2mm}
 	
(2)  Here, we repeat the argument used in the proof of Theorem \ref{theorem 1}. 
The proof is identical, 
and the only modification comes from justifying the
uniform bound $\|	 \hat f_N 	\|_{L^\infty} \leq 1 $. 
To justify this in the case $k \geq 2$, 
we note that a calculation using \eqref{f 1} and  the scaling 
$\hbar^{dk} = N^{-k }$
 	yields  the analogous Groenewold's formula \eqref{groenewold}
 	for $k$-particle Wigner functions 
 	\begin{equation}
 		\widehat{f_N^{(k)}} (   \zeta_k  	)
 		= 
 \frac{(N  - k )! }{N! }
 		\tr \big[
 		\O^{(k)}_{\zeta_k}
 		\gamma_{ N}^{(k)}
 		\big]   , \qquad \zeta_k = (\xi_k , \eta_k) \in \R^{dk}\times \R^{dk}
 	\end{equation}
 	where the trace is over $L^2(\R^{dk})$, 
 	and 
 	$	 \O^{(k)}_{\zeta_k} = \exp(  i( \xi_k \hat{X_k}	 +
 	\eta_k	 \hat{P_k})	)$
 	is a semi-classical observable 
 	in $L^2(\R^{2dk})$. 
 	Hence, 
 	thanks to  \eqref{gamma bounds} 
 	we find   
 	$\|	 \widehat{f_N^{(k)}	}	\|_{L^\infty}
 	\leq 1 	$.  
 	
 	\vspace{2mm}
 	
 	(3) We repeat the argument used in the proof of Theorem \ref{theorem 1}. 
 	Namely, we   prove  the $q =2 $ case. 
 	The only modification comes from the uniform 
 	$L^2$ bound
 	$\|	 f_N	\|_{L^2} \leq (2\pi)^{d/2} $.
For $ k \geq 2 $, 
this uniform bound   is now provided by Lemma \ref{lemma L2}, i.e. 
there holds 
$\sup_{N \geq 1} \|	f_N^{(k)}	\|_{L^2} < \infty$. 
 	We can then interpolate between $q=2$ and $q=\infty$. 
 	
 	\vspace{2mm}
 	
 	(4) It suffices to use the triangle inequality, 
 	 Lemma \ref{lemma L2}
 	and $\|	 f^{\otimes k }	\|_{L^2}	= \|	f	\|_{L^2}^k = (2\pi)^{dk/2}		$
 	to find that 
 	\begin{equation}
\liminf_{N\rightarrow \infty}
 		\|	 f_N^{(k)} - f^{\otimes k }\	\|_{L^2}
 		\geq 
\liminf_{N\rightarrow \infty} 		 		\|	 f_N^{(k)}  \|_{L^2}			-  \|	f^{\otimes k } 	\|_{L^2}
 		 		= 
 		 		(2\pi)^{dk/2} 
 		 		\Big(
 		 		\sqrt{k!} -1 
 		 		\Big) \ . 
 	\end{equation}
  This finishes the proof of the theorem. 
 	 \end{proof}

		\begin{proof}[Proof of Lemma \ref{lemma L2}]
Fix  $ k \geq 2$ and  let $N  >  k$. 
We split the proof into two parts. 
In the first part, we  calculate a relation between $\|	f_N^{(k)}\|_{L^2}$
and
$\|		 \gamma_{N}^{(k)}	\|_{HS}$, which is independent
of $\Psi_N$ being a Slater determinant.
In the second part,  we use the fact that $\Psi_N$ is a Slater determinant
to compute  the Hilbert-Schmidt norm of $\gamma_{ N}^{(k)}$. 

\vspace{1mm}
For the first part, 
we use   \eqref{f 1}
and then \eqref{f 2}
to find that 
			\begin{align}
\nonumber 
				\|	f_N^{(k)}	\|_{ L^2 } 
& 				 \  =  \ 
N^k \ \frac{(N-k)!}{N!}		\|		W_k^\hbar 	 [\gamma_N^{(k)}]\|_{L^2}			 \\
\nonumber 
&   \ =  \ 
				 N^k \ \frac{(N-k)!}{N!}
				 			 (2\pi\hbar  )^{ \frac{dk}{2}}  
				 			\|		 \gamma_N^{(k) }	\|_{L^2}	\\ 
& 				 			 \  =  \ 
				 			 N^k \ \frac{(N-k)!}{N!}
				 		(2\pi\hbar  )^{ \frac{dk}{2}}  
				 		\|		 \gamma_N^{(k) }	\|_{HS}
				 		\label{f HS}
			\end{align}
where  in the last line we   used $\|	 \gamma_k 	\|_{HS } = \|	\gamma_k 	\|_{L^2}$.

\vspace{2mm}

For the second part, we observe that  when $\Psi_N$ is  a Slater determinant, 
the $k$-particle reduced density matrix  $\gamma_N^{(k)}$ can be calculated explicitly.
Namely, assume that
\begin{equation}
	\Psi_N (x_1 , \ldots, x_N)	 =	
		 \frac{1}{\sqrt{N!}} 
		 \det_{1  \leq i,j \leq N }	
	\big[
	\vp_i(x_j)
	\big]
	 \ , \qquad (x_1, \ldots, x_N) \in \R^{dN} \ . 
\end{equation}
The orbitals $ (\vp_i)_{i=1}^N$ can depend on $N$ but we do not display such dependence in the notation, for it has no effect in the upcoming calcuation. 
Next, we note that the kernel of $\gamma_{ N}^{(k)}$ is given by  (see for reference  \cite[Section 3.1.5]{LiebSeiringer})
\begin{align}
\nonumber 
	\gamma_N^{ (k)}	 (x_1 , \ldots, x_k , x_1' , \ldots, x_k')	 
& 	   =  
	 k! 
	 \sum_{1  \leq \ell_1 < \cdots < \ell_k \leq N } 
	 \frac{1}{\sqrt k } 
	 \det_{1  \leq i,j \leq N }	
	 \big[
	 \vp_{  \ell_i 	}	(x_j)
	 \big]
	 \	
	 \overline{ \frac{1}{\sqrt k } 
	 \det_{1  \leq i,j \leq N }	
	 \big[
	 \vp_{  \ell_i 	}	(x_j'	)
	 \big]	}  \ . 
 \end{align}
Hence, we write  in operator form  
\begin{align} 
\label{eq gamma}
	\gamma_N^{ (k)}	  = 
   k! 
 \sum_{1  \leq \ell_1 < \cdots < \ell_k \leq N } 
  \ket{  \vp_{\ell_1} \wedge \cdots \wedge \vp_{\ell_k}		}
  \bra{  \vp_{\ell_1} \wedge \cdots \wedge \vp_{\ell_k}		} \ . 
\end{align}
In particular, note that $\tr \gamma_{ N}^{(k)} =  k! \binom{N}{k} = \frac{N!}{(N-k)!}	. $
Furthermore, it follows from \eqref{eq gamma}
that $\gamma_{ N}^{(k)} \gamma_{ N}^{(k)} = k! \gamma_{ N}^{(k)}$.
Thus, we can compute the Hilbert-Schmidt norm as follows 
\begin{align}
\label{gamma HS}
	\|	\gamma_N^{(k)}	 	\|_{HS}^2
		\,  =  \, 
	 \tr \Big[	( \gamma_N^{(k)}	)^* \gamma_N^{(k)}	\Big] 
 \, 	  =  \, 
	  k!  \  \tr\gamma_N^{(k)}	
	\,  = 		\,		 k! 	 \frac{N!}{(N-k)!} \  , 
\end{align}
where we have used self-adjointness of $\gamma_N^{(k)}	. $

\vspace{2mm}

Let us now put everything together. 
Namely \eqref{gamma HS} and \eqref{f HS} imply 
\begin{align}
\nonumber 
	\|		f_N^{(k)}\|_{L^2}
	&  = 
	 		 N^k \ 
	 		  \frac{(N-k)!}{N!}
	 (2\pi\hbar  )^{ \frac{dk}{2}}  
 \Big( 	  k! 	 \frac{N!}{(N-k)!}  \Big)^{\frac{1}{2}}	\\ 
 &  = 
 \Big(	 (2\pi )^{dk}  k! 	\Big)^{\frac{1}{2}}
 \bigg[
 N^k \hbar^{\frac{dk}{2}} 
\Big(  \frac{(N-k)!}{N!}\Big)^{\frac{1}{2}}
 \bigg] \ . 
 \label{f L2}
\end{align}
It suffices now to use the scaling $\hbar^d = N^{-1}$, 
take the limit $N \rightarrow \infty$
and observe that the factor in squared brackets $[\cdots]$
in \eqref{f L2} converges to $1$. 
	\end{proof}

\noindent 	\textbf{Acknowledgements.}
I am   grateful to J.K Miller and N. Pavlovi\'c
	for discussions that ultimately led to the study of the problem in this article. 
	I am also very  thankful to D. Hundertmark
	for his comments  that helped improve an earlier version of this manuscript. 
	I would also like to acknowledge 
	  important remarks from two anonymous referees
	that helped  significantly improve the conclusion of Theorem \ref{theorem 1}, 
	covering now the additional range $ p \in (2 , \infty)$. 
The author gratefully acknowledges support from the Provost’s Graduate Excellence Fellowship at The University of Texas at Austin and from the NSF grant DMS-2009549, and the NSF grant DMS-2009800 through T. Chen.

\vspace{3mm}

\noindent \textbf{Data availability.} This manuscript has no associated data.
\vspace{3mm}

\noindent 
\textbf{Conflict of interest.} 
On behalf of all authors, the corresponding author states that there is no conflict of interest.

\end{document}